\documentclass[letterpaper,10 pt,conference]{ieeeconf}  
\IEEEoverridecommandlockouts                              
\usepackage{graphics} 
\usepackage{graphicx,color}
\usepackage[active]{srcltx}
\usepackage{amsmath} 
\usepackage{amssymb}  
\usepackage{amsmath,amssymb,amsfonts,theorem}
\usepackage{enumerate}
\usepackage{tikz}
\usepackage{tabu}
\usepackage{algpseudocode} 
\usepackage{accents}
\usepackage{cite}

\usepackage{epstopdf}

{ \theorembodyfont{\normalfont} 

\newcommand{\R}{\mathbb{R}}

\renewcommand{\S}{\mathbb{S}}    

\newcommand\ie{\emph{i.e.}}

{

\newtheorem{remark}{Remark}
}

\newtheorem{definition}{Definition}
\newtheorem{theorem}{Theorem}

\newtheorem{lemma}{Lemma}

\newtheorem{corollary}{Corollary}

\def\QEDclosed{\mbox{\rule[0pt]{1.3ex}{1.3ex}}} 
\def\QEDopen{{\setlength{\fboxsep}{0pt}\setlength{\fboxrule}{0.2pt}\fbox{\rule[0pt]{0pt}{1.3ex}\rule[0pt]{1.3ex}{0pt}}}}
\def\QED{\QEDopen}

\def\be{\begin{equation}}
\def\ee{\end{equation}}
\def\ba{\begin{array}}
\def\ea{\end{array}}
\def\eqa{\begin{eqnarray}}
\def\eqe{\end{eqnarray}}
\title{\LARGE \bf Stochastic phase-cohesiveness of discrete-time Kuramoto\\oscillators in a frequency-dependent tree network}
\author{Matin Jafarian, Mohammad H. Mamduhi, Karl H. Johansson
\thanks{The authors are with the Division of Decision and Control Systems, School of Electrical Engineering and Computer Science, KTH Royal Institute of Technology, Stockholm, Sweden. Email: {\tt\small $\{$matinj,mamduhi,kallej$\}$@kth.se.}}
\thanks{This work was supported by the Knut and Alice Wallenberg Foundation, the Swedish Strategic Research Foundation and the Swedish Research Council.}}
\begin{document}
\maketitle
\thispagestyle{empty}
\pagestyle{empty}
\begin{abstract}
This paper presents the notion of {\em stochastic phase-cohesiveness} based on the concept of recurrent Markov chains and studies the conditions under which a discrete-time stochastic Kuramoto model is phase-cohesive. It is assumed that the exogenous frequencies of the oscillators are combined with random variables representing uncertainties. A bidirectional tree network is considered such that each oscillator is coupled to its neighbors with a coupling law which depends on its own noisy exogenous frequency. In addition, an undirected tree network is studied. For both cases, a sufficient condition for the common coupling strength ($\kappa$) and a necessary condition for the sampling-period are derived such that the stochastic phase-cohesiveness is achieved. The analysis is performed within the stochastic systems framework and validated by means of numerical simulations.     
\end{abstract}
\section{Introduction}\label{sec:int}
Synchronization is among the key collective behavior of many complex networks, including biological and neural networks. The well-celebrated Kuramoto model \cite{kuramoto2012chemical,strogatz2000kuramoto} has been a paradigm for studying interconnected oscillators. 

Kuramoto network has been considered in both continuous-time and discrete-time settings. Considering the continuous-time deterministic dynamics, the current literature has addressed various problems \cite{jadbabaie2004stability,
dorfler2014synchronization,mallada2013synchronization,jafarpour2018synchronization,scardovi2007synchronization,franci2010phase}, for instance conditions on the critical coupling for phase and frequency synchronization \cite{jadbabaie2004stability,dorfler2014synchronization}. 
The problem of explosive synchronization in large scale networks has also been studied using Kuramoto model by incorporating a  correlation between the structural and the dynamical properties of the network \cite{gomez2011explosive} as well as assuming coupling strength as a function of exogenous frequencies \cite{zhang2013explosive}.  

Besides the continuous-time analysis, discrete-time synchronization is another interesting direction specially that estimation of the behavior of natural/man-made systems are mainly done in a discrete-time fashion. Deterministic discrete-time Kuramoto models have been studied in e.g. \cite{klein2008integration,cenedese2016multi}. A bound for the product of coupling term and sampling period has been presented in \cite{klein2008integration} in order to achieve phase-synchronization, which is a specific form of phase-cohesiveness, where the underlying graph has been either a complete or a star graph with a common and constant coupling strength and zero exogenous frequencies. 

In addition to deterministic models, stochastic Kuramoto network has also been studied using a continuous-time Fokker-Planck model to analyze the network behavior where a noise was added to the dynamics of each oscillator, e.g., \cite{acebron2005kuramoto,bag2007influence,breakspear2010generative}. Also, Markov chains have been utilized for the analysis of the effects of phase discretization on the synchronization \cite{jorg2017stochastic}. 

Kuramoto model has been widely used to study synchronization in several disciplines including neural networks \cite{breakspear2010generative,cumin2007generalising}. The model has been utilized to study the behavior of spiking neurons represented by conductance-based models \cite{izhikevich2007dynamical}. In such models, the coupling conductance can be subject to fluctuations due to noises \cite{richardson2005synaptic}.
 
Motivated by noisy interconnections in neural networks as well as considering the explosive synchronization behavior modeled by frequency-dependent coupling, the main contribution of the paper is to present a notion of stochastic cohesive behavior for such a network and characterize conditions under which this behavior is achieved for stochastic discrete-time Kuramoto oscillators. We consider both bidirectional frequency-dependent tree networks as well as undirected tree networks. Our choice of studying tree networks is encouraged by observations that large-scale inter-areal connectivity in the brain can be approximated as a tree network \cite{stam2014trees}. 

We use the term {\em stochastic phase-cohesiveness} to refer to a probabilistic counterpart (employing the concept of {\em recurrent Markov chains}), of the standard deterministic phase-cohesiveness. We assume that all exogenous frequencies are uncertain. The oscillators' exogenous frequencies are modeled as the sum of a constant value and a random variable. We first consider a network of oscillators with a bidirectional graph topology such that each oscillator is coupled to its neighbors with a coupling term equal to the product of a common coupling term, namely $\kappa$, and its own uncertain exogenous frequency. As a result, considering a frequency-dependent network, the weights of all edges of the graph are affected by the uncertain exogenous signals. We derive a sufficient condition on the bound of $\kappa$ and a necessary condition for the sampling-time such that the discrete-time network achieves stochastic phase-cohesiveness assuming that the expectation of the minimum eigenvalue of the graph weighted edge Laplacian is positive. Analogously, we perform the analysis for the undirected tree graphs, \ie\  all edges have a common and constant positive coupling strength $\kappa$, and provide conditions for the stochastic phase-cohesiveness. Compared to our previous work \cite{jafarian2018sync}, which studied frequency synchronization in a continuous-time bidirectional tree network with constant and positive exogenous frequencies, this paper studies a stochastic discrete-time network where each exogenous frequency is modeled as the sum of a constant-value and a random variable which may also take negative values. Thus, the weights of the graph edges are not always positive. The latter makes the analysis of such a network more challenging. We use a probabilistic measure, the expectation of the minimum eigenvalue of the weighted Laplacian, to tackle this problem. 
 
To the best of our knowledge, stability of uncertain parameter Kuramoto models, where the uncertainty has a random nature, have not been considered in the existing literature. In particular, the case of frequency-dependent networks where the interconnection of each two oscillators is subject to random uncertainties has not been studied. 
 
The paper is organized as follows. Section \ref{sec:pre} presents preliminaries and problem formulation. Section \ref{sec:con} presents the notion of stochatic phase-cohesiveness and studies the conditions under which this behavior is achieved for a frequency-dependent network. 
The analysis of the undirected graph is presented in Section \ref{sec:comp}. Section \ref{sec:sim} presents simulation results and Section \ref{sec:cl} concludes the paper.\\ 
\noindent{\bf{Notation}:}
Symbol $\mathbf{1}_n$ is a $n$-dimensional vector. The empty set is denoted by $\emptyset$. The notations $x_{i,j}$ and $x_k$ are equivalently used for $x_i-x_j$ and $x(\mathrm{k})$, respectively. The symbol $\S^1$ denotes the unit circle. The term $|\theta_i-\theta_j|$ denotes the {\em geodesic distance} between two angles $\theta_i \in \S^1$, $\theta_j \in \S^1$ defined as the minimum of the counter-clockwise and the clockwise arc lengths connecting $\theta_i$ and $\theta_j$ \cite{dorfler2014synchronization}. A random variable $x$ selected from an arbitrary distribution $\mathcal{X}$ with mean $\mu$ and variance $\Sigma$ is denoted by $x \sim \mathcal{X} (\mu, \Sigma)$. The expected value and conditional expected value operators are denoted by $\boldsymbol E[\cdot]$ and $\boldsymbol E[\cdot|\cdot]$, respectively. 
\section{Preliminaries and Problem Statement}\label{sec:pre}
In this section, we first revisit some preliminaries of the graph theory and Markov chains (MC), and then we state the problem formulation.\\[1mm]
{\bf Graph theory preliminaries:} For a connected undirected graph $G(\mathcal V,\mathcal E)$, the node-set $\mathcal V$ corresponds to $n$ nodes and the edge-set $\mathcal E \!\subset \!\mathcal V \!\times \!\mathcal V$ corresponds to $m$ edges. The incidence matrix $B_{n \times m }$ associated to $G(\mathcal V,\mathcal E)$ describes which nodes are coupled by an edge. 
The matrix $L= B B^T$ is called the graph Laplacian and $L_g=B^T B$ is the edge Laplacian. 
For any undirected tree graph, all eigenvalues of $L_g$ are equal to the nonzero eigenvalues of $L$ \cite{mesbahi2010graph}. In this paper, we consider  trees which are a  subclass of connected graphs without cycles, i.e., any two nodes are connected by exactly one unique path. The edge Laplacian of a tree graph is invertible \cite{mesbahi2010graph}.\\[1mm]
{\bf Stochastic processes:}
The three-tuple $(\Omega, \mathcal{F}, \textsf{P})$ defines a probability space, where $\Omega$ (sample space) is the set of all possible outcomes, $\mathcal{F}$ is a $\sigma$-algebra\footnote{A $\sigma$-algebra $\mathcal{F}$ defined on a set $\Omega$ is a set containing subsets of $\Omega$ including the empty set.} of events with associated probabilities determined by the probability measure $\textsf{P}$. The following definitions are mainly borrowed from \cite{meyn2012markov}. 
\begin{definition} \cite[Ch.3]{meyn2012markov}
Let $\Omega$ be a sample space, and $\mathcal{F}$ any $\sigma$-algebra on $\Omega$, i.e. the pair $(\Omega, \mathcal{F})$ is a measurable space. 
We call $\Omega$ an \textit{uncountable} space if it is assigned a countably generated $\sigma$-algebra\footnote{Assume $\mathcal{B}$ is a random process (arbitrary family of subsets of $\Omega$) defined on a probability space $(\Omega, \mathcal{F} , \mathsf{P})$. Then, the smallest $\sigma$-algebra on which $\mathcal{B}$ is measurable, i.e. the intersection of all $\sigma$-algebras on which $\mathcal{B}$ is measurable, is called the generated $\sigma$-algebra by $\mathcal{B}$.} $\mathcal{F}(\Omega)$.
\end{definition}

We say $(\Omega,\mathcal{F})$ is a measurable space if the $\sigma$-algebra $\mathcal{F}$ on $\Omega$ satisfies the following properties:
(a) $\emptyset \in \mathcal{F}$,
(b) If $B \in \mathcal{F}$, then $B^c \in \mathcal{F}$, where $B^c=\Omega\setminus B$,
(c) If $B_1\in \mathcal{F}$ and $B_2\in \mathcal{F}$, then $B_1\cup B_2 \in \mathcal{F}$. 

The probability measure $\mathsf{P}:\mathcal{F}\rightarrow [0,1]$ is a measure on $(\Omega, \mathcal{F})$ that assigns a probability to each outcome of~$\mathcal{F}$.
\begin{definition}\label{def:markov}
A stochastic process $\Phi=\{\Phi_0, \Phi_1, \ldots\}$ evolving on a sample space $\Omega$ associated by a probability law $\mathsf{P}$ is a \textit{time homogeneous Markov chain} if for some sets $\mathcal{B}$ a set of transition probabilities $\{P^n(\omega,\mathcal{B}), \omega\in \Omega, \mathcal{B}\subset \Omega\}$ exist such that for $n,m$ in $\mathbb{Z}^+$
\begin{equation*}
\mathsf{P}(\Phi_{n+m}\in \mathcal{B}|\Phi_j,j\leq m, \Phi_m=\omega)=P^n(\omega,\mathcal{B}).
\end{equation*}
The independence of the transition probability $P^n(\omega,\mathcal{B})$ from $j\leq m$ is the Markov property, and its independence from $m$ is the time homogeneity property.
\end{definition}

\begin{definition}\label{def:irreducibility}
Let a MC $\Phi\!=\!\{\Phi_0,\Phi_1,\ldots\}$ evolves in general sample space $\Omega$ equipped with $\sigma$-algebra $\mathcal{F}(\Omega)$. Then:
\begin{enumerate}
\item for any $B\in \mathcal{F}(\Omega)$, the measurable function $\tau_B: \Omega \rightarrow \mathbb{Z}^+\cup \{\infty\}$ denotes the first return time to the set $B$, i.e.
\begin{equation}
\tau_B \triangleq \min \{n\geq 1\;|\;\Phi_n\in B\}.
\end{equation}
\item for any measure $\varphi$ on the $\sigma$-algebra $\mathcal{F}(\Omega)$, $\Phi$ is said to be $\varphi$-irreducible if $\forall \; \omega \in \Omega$, and  $B\in \mathcal{F}$, $\varphi(B)>0$ implies
$\mathsf{P}_{\omega}(\tau_B<\infty)>0$.
 \end{enumerate}
 \end{definition}
According to the definition \ref{def:irreducibility}, the entire state space of a MC is reachable, independent of the initial state, via finite number of transitions only if the MC is $\varphi$-irreducible. Moreover, if a MC is $\varphi$-irreducible, then a unique maximal irreducibility measure $\psi > \varphi$ exists on $\mathcal{F}(\Omega)$ such that $\Phi$ is $\varphi^\prime$-irreducible for any other measure $\varphi^\prime$ if and only if $\psi > \varphi^\prime$. We say then the MC $\Phi$ is $\psi$-irreducible.
\subsection{Problem statement}\label{sec:pf}
We consider $n$ discrete-time oscillators communicating over a connected and bidirectional tree graph such that each oscillator dynamics is obtained with a first order (zero-order hold) discrete-time approximation of 
\be\label{eq:m}
{\dot\theta}_i= (\omega_i + n_i(t)) (1 - \kappa \sum_{j \in {\cal N}_i} \sin(\theta_i(t)-\theta_j(t))), 
\ee
such that the discrete-time dynamics of oscillator $i$ is {\footnote{The model is motivated by a frequency-dependent neural network subject to fluctuations in coupling strength (see Section \ref{sec:int}).}} 
\begin{equation}\label{eq:m1h}
{\theta_i}(\mathrm{k}+1)\!=\! {\theta_i}(\mathrm{k})+ \tau (\omega_i+ n_i(\mathrm{k})) (1 - \kappa \!\sum_{j \in {\cal N}_i} \!\sin(\theta_i(\mathrm{k})-\theta_j(\mathrm{k}))), 
\end{equation}
where $k >0$, $\theta_i(\mathrm{k}) \in \S^1$, $\omega_i \in \R$, and $\tau >0$ represent the time step, the phase and exogenous frequency of oscillator $i$, and the sampling time, respectively. Symbol ${\cal N}_i$ denotes the set of neighbors of node $i$. The parameter $\kappa > 0, \kappa \in \R$ is the constant coefficient of the coupling strength of all links of the graph. The disturbance process $n_i$, for all $i\in\{1,\ldots,n\}$, is assumed to be an i.i.d random sequence with the random realizations $n_i(\mathrm{k})$ selected from an arbitrary continuous distribution with finite mean and variance, at each $k$. This model can be interpreted as a bidirectional communication where the weights of coupling of each edge at each direction depends on the randomly disturbed exogenous frequency of its head oscillator.\\ 
We define the augmented phase state $\boldsymbol\theta(\mathrm{k})\triangleq [\theta_1(\mathrm{k}),\ldots,\theta_n(\mathrm{k})]^\top$. The compact form of the relative phase dynamics can then be written as
\be\begin{aligned}\label{eq:rel}
B^T \boldsymbol\theta(\mathrm{k}+1)&= B^T \boldsymbol\theta(\mathrm{k}) + \tau B^T {\underaccent{\sim}{\boldsymbol \omega}}(\mathrm{k}) {\boldsymbol 1}_n \\&- \kappa \tau B^T {\underaccent{\sim}{\boldsymbol \omega}}(\mathrm{k}) B \sin(B^T \boldsymbol\theta(\mathrm{k})), 
\end{aligned}\ee
where $\boldsymbol\theta(\mathrm{k})$ is the state vector at time step $k$, $B$ is the graph incidence matrix and ${\underaccent{\sim}{\boldsymbol \omega}}(\mathrm{k})$ is a diagonal random matrix at time step $k$ whose diagonal elements are equal to the noisy exogenous frequency of the oscillators, \ie, 
\be\label{eq:ran}
{\underaccent{\sim}{\boldsymbol \omega}}_{n \times n}(\mathrm{k})=
  \begin{pmatrix}
    \omega_{1}+ n_1(\mathrm{k})& \dots & 0\\
    \vdots & \!\!\!\ddots\!\!\! & \vdots \\
    0 & \dots & \omega_{n}+ n_n(\mathrm{k})\\
  \end{pmatrix}.
\ee 
We also assume that the initial relative phase, \ie\  $\theta_{i}(0)-\theta_{j}(0)$, is an arbitrary random variable, independent from the noise realizations $n_i(\mathrm{k})$, selected from any finite moment probability distribution with continuous density function such that $|\theta_{i}(0)-\theta_{j}(0)| \leq \gamma, j \in {\cal N}_i$, where $\gamma>0, \gamma=\frac{\pi}{2}-\varepsilon$ for some $\varepsilon>0$, and $|.|$ denotes the geodesic distance \cite{dorfler2014synchronization}. Thus, the set of randomly selected initial conditions is
\begin{equation}\label{eq:s}
S^{G}(\gamma)\!=\!\left\{\!\theta_i(\mathrm{k}) \!\in \S^1\!, |\theta_{i,j}(\mathrm{k})| \leq \gamma, 0<\!\gamma < \!\frac{\pi}{2}\right\}\!.
\end{equation}

The noise process $n_i$'s for all $i\in\{1,\ldots,n\}$ together with the random initial relative phase $\theta_{i,j}(0)$ for all $i$ and $j\in \mathcal{N}_i$ generate a probability space $(\Omega, \mathcal{F}, \textsf{P})$ where $\Omega$ (sample space) is the set of all possible outcomes, $\mathcal{F}$ is a $\sigma$-algebra of events with the associated probabilities determined by the probability measure function $\textsf{P}$. It is straightforward to conclude that the three-tuple $(\Omega, \mathcal{F}, \textsf{P})$ represents an uncountable probability space, because the noise distribution has a continuous density function and the noise realizations can take any values of their supporting range $(-\infty,+\infty)$ at any time instance $k$. 

According to (\ref{eq:rel}), dynamics of $\boldsymbol\theta(\mathrm{k}+1)$ depends only on the most recent state $\boldsymbol\theta(\mathrm{k})$ and the noise variables $\{n_1(\mathrm{k}),\ldots,n_n(\mathrm{k})\}$, therefore, we can conclude that $\boldsymbol\theta(\mathrm{k})$ is a Markov chain (MC) for all $k\in\{0,1,2,\ldots\}$ with its dynamics evolving over the uncountable probability space $(\Omega, \mathcal{F}, \textsf{P})$. Based on the definitions \ref{def:markov} and \ref{def:irreducibility}, $\boldsymbol\theta(\mathrm{k})$ is a time homogeneous MC because the difference equation (\ref{eq:rel}) is time-invariant and the noise process $n_i$'s are i.i.d. for $i \in \{1, \ldots, n\}$ at every time-step $k$. This implies that $\boldsymbol\theta(\mathrm{k})$ evolves according to a stationary transition probability on the sample space. Moreover, the chain is a $\psi$-irreducible, with $\psi$ a nontrivial measure on the $\sigma$-algebra, since the noise distribution is absolutely continuous having a positive density function at any state of the probability space.
\begin{definition}\cite{meyn1994}\label{def:PHR}
Let the $\psi$-irreducible MC $\Phi=\{\Phi_0,\Phi_1,\ldots\}$ be defined over the probability space $(\Omega,\mathcal{F},\textsf{P})$ with random variables measurable with respect to some known $\sigma$-algebra $\mathcal{F}(\Omega)$. Then $\Phi$ is said to be \textit{recurrent} if
\begin{equation*}
\textsf{P}_\omega(\tau_B<\infty)=1, \forall \omega\in \Omega, B\in \mathcal{F}.
\end{equation*}
\end{definition}
Intuitively, definition \ref{def:PHR} states that if a state of a recurrent MC leaves a subset $B \in \mathcal{F}$ with non-zero probability, the state returns to the set $B$ with probability one.\\
Consider $0< \delta < \frac{\pi}{2}$. A deterministic phase-cohesiveness implies that $\forall i, j \in \mathcal V, |\theta_j(\mathrm{k})-\theta_i(\mathrm{k})|< \delta$ \cite{dorfler2014synchronization}. Inspired by this, we now define the stochastic phase-cohesiveness. First, let us define
\begin{equation}\label{eq:si}
\Omega=\left\{\theta_i(\mathrm{k}) \in \S^1, |\theta_{i,j}(\mathrm{k})| \leq \frac{\pi}{2} \right\}.
\end{equation}
\begin{definition}\label{def1}
The relative phase process $B^T \boldsymbol\theta$ defined on the probability space $(\Omega,\mathcal{F},\textsf{P})$ with $\Omega$ in \eqref{eq:si} is phase-cohesive in a stochastic sense if the process is recurrent having a desired subset of the state space, namely $S^{G}(\gamma)$ defined in \eqref{eq:s}, \ie\ $\textsf{P}_{ B^T \boldsymbol\theta(\mathrm{k})}(\tau_{S^{G}(\gamma)} <\infty)=1, \forall k, \forall \ \boldsymbol\theta(\mathrm{k}) \in \Omega, S^{G}(\gamma) \in \mathcal{F}$.
\end{definition}
{\bf{Problem}} 
Our goal is to study conditions under which the stochastic phase-cohesiveness can be achieved for the relative-phase process in \eqref{eq:rel}. In \cite{jafarian2018sync}, a deterministic, continuous-time and noise-free counterpart of \eqref{eq:rel} has been analyzed. It was shown that for a sufficiently large $\kappa$, \ie 
\be\label{eq:k}
\kappa > \frac{|\omega_{\max} -\omega_{\min}|}{\lambda_{\min}(B^T {\boldsymbol \omega} B) \sin(\gamma)},
\ee
with $\boldsymbol \omega$ being a noise-free positive diagonal matrix with a similar structure as in \eqref{eq:ran}, the deterministic, continuous-time and noise-free network achieves phase-cohesiveness.\\ 
In this paper, we consider the case where the exogenous frequencies, and hence the weights of edges, are combined with random uncertainties which are not always positive, hence, contrary to \cite{jafarian2018sync}, this paper allows random zero and negative edge weights as well. 
Notice that $\theta_i(\mathrm{k}) \in \S^1, \forall k$ indicates that the dynamics is evolving on the $n$-Torus \cite{scardovi2007synchronization}. In this paper, we study the relative process \eqref{eq:rel} and restrict the randomly selected initial conditions to $S^{G}(\gamma)$ which is diffeomorphic to the Euclidean space and hence run our analysis in the Euclidean space.  
\section{Stochastic phase cohesiveness: Bidirectional asymmetric network}\label{sec:con}
In this section, we derive a sufficient coupling condition and a necessary sampling-time condition to achieve phase-cohesiveness for stochastic oscillators modeled in \eqref{eq:rel}. Our results are based on the concept of recurrent MC for discrete-time stochastic systems \cite{meyn2012markov}. The following lemma which is directly inferred from Theorem 8.4.3 of \cite{meyn2012markov} provides the condition under which a $\psi$-irreducible MC is recurrent.
\begin{lemma}\label{thm:non-evanescent}
Assume that $\Phi=\{\Phi_0, \Phi_1,\ldots\}$ is a $\psi$-irreducible MC evolving on a sample space $\Omega$ equipped with the $\sigma$-algebra $\mathcal{F}$, and $V(\Phi):\Omega \rightarrow \mathbb{R}^+$ is a monotonic function such that $V(\Phi)\rightarrow \infty$ as $\Phi \rightarrow \infty$. The MC $\Phi$ is recurrent if there exists a small set{\footnote{Theorem 8.4.3. of \cite{meyn2012markov} requires a {\it petite} set. Every small set is a petite set and small sets always exist for $\varphi$-irreducible chains \cite{meyn2012markov}.}} $C\in \mathcal{F}$ satisfying
\begin{equation*}
\boldsymbol E[V(\Phi_{k+1})|\Phi_k=x]-V(x)< 0, \quad \;\forall \Phi_k \in \Omega \setminus C.
\end{equation*}
\end{lemma}
\begin{theorem}\label{pr1}
Consider the discrete-time stochastic process in \eqref{eq:rel} with the set $S^{G}(\gamma)=\left\{\theta_i(\mathrm{k}) \in \S^1, |\theta_{i,j}(\mathrm{k})| \leq \gamma, 0<\gamma < \frac{\pi}{2}\right\}$. Assume that ${\boldsymbol E}[\lambda_{\min}(B^T {\underaccent{\sim}{\boldsymbol \omega}} B)]$ is strictly positive. Then, the process \eqref{eq:rel} is recurrent if the two following conditions hold
\be\begin{aligned}\label{eq:kappa}
& \kappa > \frac{(1-\sin(\gamma)) \frac{\pi}{2 \tau} + {\boldsymbol E}_{\max}|\Delta{\underaccent{\sim}{\omega}}|} {(\sin^{2}(\gamma)) {\boldsymbol E} [\lambda_{\min}(B^T {\underaccent{\sim}{\boldsymbol \omega}} B)]}\\
& \tau  < \frac{(1+\sin(\gamma)) \gamma}{\kappa {\boldsymbol E}[\lambda_{\max}(B^T {\underaccent{\sim}{\boldsymbol \omega}} B)]+ {\boldsymbol E}_{\max}|\Delta{\underaccent{\sim}{\omega}}|},
\end{aligned}\ee
where ${\boldsymbol E}_{\max}|\Delta{\underaccent{\sim}{\omega}}|= {\underaccent{i,j}{\max}}{\boldsymbol E}[|(\omega_i+n_i(\mathrm{k})) - (\omega_j+n_j(\mathrm{k}))|]$.
\end{theorem} 
 
\begin{proof}
We prove \eqref{eq:rel} is recurrent w.r.t. $S^{G}(\gamma)$ by taking radially unbounded function $V= (\sin \gamma){\boldsymbol 1}_m^T |B^T \boldsymbol\theta|= \sin(\gamma) \sum_{i,j} |\theta_i(\mathrm{k})-\theta_j(\mathrm{k})|$ and showing that the one-step drift of V, \ie,  
\be\begin{aligned}\label{eq:dv1}
\Delta V (\boldsymbol \theta)=&{\boldsymbol E} [V(\boldsymbol\theta({k+1}))|\boldsymbol\theta({k})]- V(\boldsymbol\theta({k})),\\
=&{\boldsymbol E} [\sin(\gamma) {\boldsymbol 1}_m^T |B^T \boldsymbol\theta({k+1})|\ |\ \boldsymbol\theta({k})] - \\&\sin(\gamma) {\boldsymbol 1}_m^T |B^T \boldsymbol\theta(\mathrm{k})|,
\end{aligned}\ee
is negative if $\left\{\forall i,j,\quad \gamma \leq |\theta_i(\mathrm{k})-\theta_j(\mathrm{k})| < \frac{\pi}{2}\right\}$. 
Consider the relative phase dynamics as in \eqref{eq:rel} and let $I^k$ denote the right-hand side of equation \eqref{eq:rel}. Then $(\sin \gamma) {\boldsymbol 1}_m^T |B^T \boldsymbol\theta(\mathrm{k}+1)|= (\sin \gamma){\boldsymbol 1}_m^T |I^k|$. Notice that for $\gamma \leq |\theta_i-\theta_j| < \frac{\pi}{2}$, we have $\sin(\gamma) \leq |\sin(\theta_{i,j})| <1, \sin(\gamma)>0$. Thus, 
\be\label{eq:ie1}
(\sin \gamma) {\boldsymbol 1}_m^T |I^k| \leq |\sin^T(B^T \boldsymbol \theta(\mathrm{k}))||I^k| < {\boldsymbol 1}_m^T |I^k|.
\ee
Moreover, we have
\be\begin{aligned}\label{eq:e2}
|\sin^T(& B^T \boldsymbol \theta(\mathrm{k}))||I^k| =|\sin^T(B^T \boldsymbol \theta(\mathrm{k}))I^k|= \\&|\sin^T(B^T \boldsymbol \theta(\mathrm{k})) B^T \boldsymbol \theta(\mathrm{k}) + \tau \sin^T(B^T \boldsymbol \theta(\mathrm{k})) B^T {\underaccent{\sim}{\boldsymbol \omega}}(\mathrm{k}) {\boldsymbol 1}_n \\ &- \kappa \tau \sin^T(B^T \boldsymbol \theta(\mathrm{k})) (B^T {\underaccent{\sim}{\boldsymbol \omega}}(\mathrm{k}) B) \sin(B^T \boldsymbol\theta(\mathrm{k}))|. 
\end{aligned}\ee
Denote $\sin(\gamma) {\boldsymbol 1}_{m}^T$ by $\boldsymbol 1^T_{\gamma}$. We can write
\be\begin{aligned}\label{eq:ie2}
\sin(B^T \boldsymbol\theta(\mathrm{k})) B^T \boldsymbol\theta(\mathrm{k}) &< ({\boldsymbol 1^T_{\gamma}}+({\boldsymbol 1^T}-{\boldsymbol 1^T_{\gamma}})) |B^T \boldsymbol\theta(\mathrm{k})|\\ & < {\boldsymbol 1^T_{\gamma}} |B^T \boldsymbol\theta(\mathrm{k})|+ ({\boldsymbol 1^T}-{\boldsymbol 1^T_{\gamma}}) |B^T \boldsymbol\theta(\mathrm{k})|.
\end{aligned}\ee
Combining \eqref{eq:ie1}, \eqref{eq:e2} and \eqref{eq:ie2}, we obtain
\be\begin{aligned}\label{eq:ie3}
\boldsymbol 1^T_{\gamma}|I^k| < &|{\boldsymbol 1^T_{\gamma}} |B^T \boldsymbol\theta(\mathrm{k})|+ ({\boldsymbol 1^T}-{\boldsymbol 1^T_{\gamma}}) |B^T \boldsymbol\theta(\mathrm{k})|+\\& \sin^T(B^T \boldsymbol\theta(\mathrm{k})) [\tau B^T {\underaccent{\sim}{\boldsymbol \omega}}(\mathrm{k}) {\boldsymbol 1}_n - \\& \kappa \tau B^T {\underaccent{\sim}{\boldsymbol \omega}}(\mathrm{k}) B \sin(B^T \boldsymbol\theta(\mathrm{k}))]|. 
\end{aligned}\ee
Now, in the view of \eqref{eq:dv1}, we can write 
\be\begin{aligned}\label{eq:dv2}
\Delta V (\boldsymbol \theta) &< {\boldsymbol E} [\ |\ \underbrace{{\boldsymbol 1^T_{\gamma}} |B^T \boldsymbol\theta(\mathrm{k})|}_{\bar a} + \underbrace{({\boldsymbol 1^T}-{\boldsymbol 1^T_{\gamma}}) |B^T \boldsymbol\theta(\mathrm{k})|}_{\bar b}-\\[1mm]& \underbrace{\kappa \tau \sin^T(B^T \boldsymbol\theta(\mathrm{k})) B^T {\underaccent{\sim}{\boldsymbol \omega}}(\mathrm{k}) B \sin(B^T \boldsymbol\theta(\mathrm{k}))}_{c}+\\[1mm]& \underbrace{\tau \sin^T(B^T \boldsymbol\theta(\mathrm{k})) B^T {\underaccent{\sim}{\boldsymbol \omega}}(\mathrm{k}) {\boldsymbol 1}_n}_{d}\ |\ ] -{\boldsymbol 1^T_{\gamma}} |B^T \boldsymbol\theta (\mathrm{k})|,
\end{aligned}\ee
where $\bar a > 0, \bar b > 0$ and $c, d \in \R$. If ${\boldsymbol E}[|\bar a + \bar b -c+d|] < \bar a$ holds, then we can conclude that $\Delta V <0$. Notice that $|\bar a + \bar b -c+d|$ is either equal to $\bar a + \bar b -c+d$ (if $\bar a + \bar b -c+d >0$) or equal to $-\bar a - \bar b +c-d$ (if $\bar a + \bar b -c+d <0$). Hence, the following should hold
\begin{enumerate}
\item ${\boldsymbol E}[\bar a + \bar b -c+d] < \bar a$ \ if \ $\bar a + \bar b -c+d >0$,
\item ${\boldsymbol E}[-\bar a - \bar b +c-d] < \bar a$\  if \ $\bar a + \bar b -c+d <0$.
\end{enumerate}
Notice that ${\boldsymbol E}[\bar a]=\bar a$ and ${\boldsymbol E}[\bar b]=\bar b$. Hence, if $-2 \bar a < \bar b+{\boldsymbol E}[-c+d]<0$ holds, then $\Delta V <0$. Since, matrix $B^T {\underaccent{\sim}{\boldsymbol \omega}}(\mathrm{k}) B$ is symmetric (with the same structure as in (11) in \cite{jafarian2018sync}), and we assume that ${\boldsymbol E}[\lambda_{\min}(B^T {\underaccent{\sim}{\boldsymbol \omega}}(\mathrm{k}) B)] >0$, 
 then we have ${\boldsymbol E}[c] > \kappa \tau {\boldsymbol E}[\lambda_{\min}(B^T {\underaccent{\sim}{\boldsymbol \omega}}(\mathrm{k}) B)] \sin^{T}(B^T \boldsymbol\theta(\mathrm{k}))\sin(B^T \boldsymbol\theta(\mathrm{k}))>0$.
Thus, the two following criteria should hold in order to obtain $\Delta V <0$:
\be\begin{aligned}\label{eq:two}
&\min\ {\boldsymbol E}[c] > \max\ \bar b+ {\boldsymbol E}[|d|],\\
&\max\ {\boldsymbol E}[c + |d|] < \min\ (2 \bar a + \bar b).
\end{aligned}\ee
Recall that we are analyzing the system assuming that $\left\{\forall i,j,\quad \gamma \leq |\theta_i-\theta_j| < \frac{\pi}{2}\right\}$. Thus, $|B^T \boldsymbol \theta(\mathrm{k})|$ and $|\sin^T(B^T \boldsymbol \theta(\mathrm{k}))|$ are upper (lower) bounded by $\frac{\pi}{2} {\boldsymbol 1}_m$ ($\gamma {\boldsymbol 1}_m$) and ${\boldsymbol 1}_m$ (${\boldsymbol 1^T_{\gamma}}$) respectively. Hence, from \eqref{eq:two} we obtain 
\be\begin{aligned}\nonumber
m \kappa \tau (\sin^{2}(\gamma)) {\boldsymbol E}[\lambda_{\min}(B^T {\underaccent{\sim}{\boldsymbol \omega}} B)] >& (1-\sin(\gamma)) \frac{m \pi}{2} + \\&m \tau {\boldsymbol E}_{\max}|\Delta{\underaccent{\sim}{\omega}}|\\\nonumber
 m \kappa \tau {\boldsymbol E}[\lambda_{\max}(B^T {\underaccent{\sim}{\boldsymbol \omega}} B)] + m \tau {\boldsymbol E}_{\max}&|\Delta{\underaccent{\sim}{\omega}}| < m (1+\sin(\gamma)) \gamma,
\end{aligned}\ee
where $m$ is the number of edges of the graph and ${\boldsymbol E}_{\max}|\Delta{\underaccent{\sim}{\omega}}|$ is the maximum expectation of disturbed relative exogenous frequencies over all edges. Finally, calculating $\kappa$ and $\tau$ from the above inequalities gives \eqref{eq:kappa} which ends the proof.
\end{proof}
\begin{remark}
Notice that the result in \eqref{eq:kappa} provides a necessary condition for the sampling time $\tau$ and a sufficient condition for the coupling strength $\kappa$. If we choose $\gamma$ sufficiently close to $\frac{\pi}{2}$, then the term $(1-\sin(\gamma))$ is negligible and the sufficient bound for $\kappa$ is comparable with \eqref{eq:k} which is obtained for the continuous-time deterministic frequency-dependent network (with a different choice of Lyapunov function) \cite{jafarian2018sync}. 
\end{remark}
\section{Stochastic phase cohesiveness: Undirected network}\label{sec:comp}
In this section, we consider an undirected network such that the coupling strength of all links is equal to $\kappa >0$. We assume that each oscillator dynamics follows 
\begin{equation}\label{eq:m2h}\nonumber
{\theta_i}(\mathrm{k}+1)= {\theta_i}(\mathrm{k})+ \tau (\omega_i+ n_i(\mathrm{k}))- \kappa \tau \sum_{j \in {\cal N}_i} \sin(\theta_i(\mathrm{k})-\theta_j(\mathrm{k})). 
\end{equation}
The compact network model is
\be\begin{aligned}\label{eq:rel2}
B^T \boldsymbol\theta(\mathrm{k}+1)=& B^T \boldsymbol\theta(\mathrm{k})+ \tau B^T {\underaccent{\sim}{\boldsymbol \omega}}(\mathrm{k}) {\boldsymbol 1}_n\\-& \kappa \tau B^T B \sin(B^T \boldsymbol\theta(\mathrm{k})).
\end{aligned}\ee
\begin{corollary}\label{co2}
Consider the discrete-time stochastic process in \eqref{eq:rel2} with the set $S^{G}(\gamma)=\left\{\theta_i(\mathrm{k}) \in \S^1,  |\theta_{i,j}(\mathrm{k})| \leq \gamma, 0<\gamma < \frac{\pi}{2}\right\}$. Then, the process \eqref{eq:rel2} is recurrent provided that the two following conditions hold
\be\begin{aligned}\label{eq:kappa2}
& \kappa > \frac{(1-\sin(\gamma)) \frac{\pi}{2 \tau} + {\boldsymbol E}_{\max}|\Delta{\underaccent{\sim}{\omega}}|}{(\sin^{2}(\gamma)) \lambda_{\min}(B^T B)}\\
& \tau  < \frac{(1+\sin(\gamma)) \gamma}{\kappa \lambda_{\max}(B^T B)+ {\boldsymbol E}_{\max}|\Delta{\underaccent{\sim}{\omega}}|},
\end{aligned}\ee
where ${\boldsymbol E}_{\max}|\Delta{\underaccent{\sim}{\omega}}|= {\underaccent{i,j}{\max}}{\boldsymbol E}[|(\omega_i+n_i(\mathrm{k})) - (\omega_j+n_j(\mathrm{k}))|]$.
\end{corollary} 
Proof of Corollary \ref{co2} follows a similar trend as the proof of Theorem \ref{pr1} and hence omitted. Note that the edge Laplacian $B^T B$ is positive definite independent of random variables. 
\begin{remark}
For a deterministic discrete-time Kuramoto model with zero-exogenous frequencies, \cite{klein2008integration} proved that $0 < n \tau \kappa <2$ should hold for a complete graph in order to achieve phase-synchronization which is a specific form of phase-cohesiveness. Assuming that both exogenous frequencies and their corresponding random disturbances are zero (hence, a similar problem setting as in \cite{klein2008integration}), for a two node graph, which is both a complete and a tree graph, the result in \cite{klein2008integration} gives $0 < \tau \kappa <1$ and Corollary \ref{co2} gives $0 < \tau \kappa < \frac{\pi}{2}$. 
\end{remark}
\section{Simulation results}\label{sec:sim}
This section presents simulation results for a network of five oscillators over a line graph. Figure \ref{g1} shows the frequency-dependent network. The initial condition for the oscillators is set to $\boldsymbol \theta(0)=[\frac{\pi}{4},\frac{\pi}{8},\frac{-\pi}{8},\frac{-\pi}{5},\frac{\pi}{5}]$. The exogenous frequencies are set to $\omega_1=7, \omega_2=10, \omega_3=1, \omega_4=6, \omega_5=2$. Gaussian random variables are considered with variances equal to $V_1=3, V_2=5, V_3=0.5, V_4=2, V_5=1$. 
\begin{figure}[h]
\centering
\includegraphics[scale=0.5]{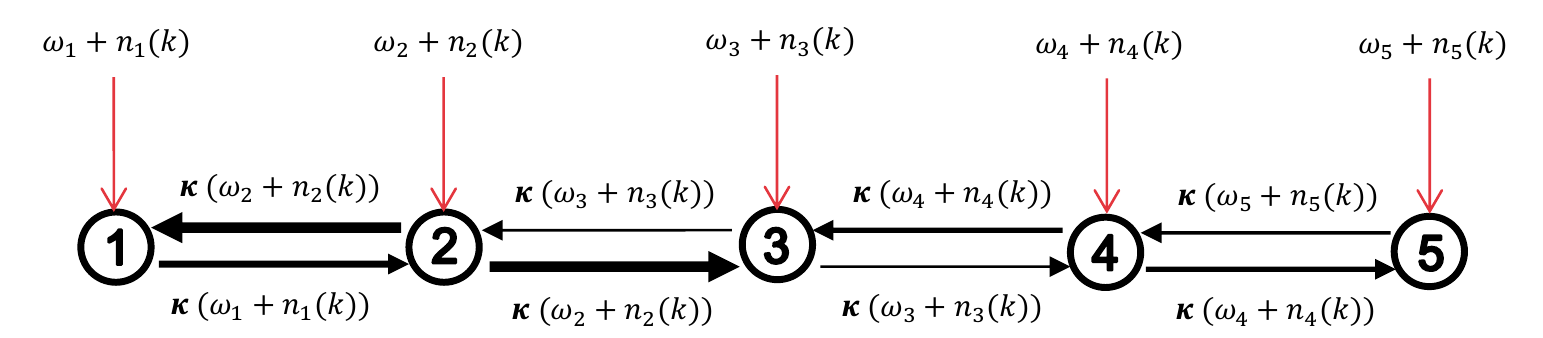}\vspace{-2mm}
\caption{Frequency-dependent line network.}\label{g1}
\end{figure}
We first simulate the network in Figure \ref{g1} assuming random variables with zero mean. The value of $\kappa$ and $\tau$ are calculated based on \eqref{eq:kappa}. The values of ${\boldsymbol E} [\lambda_{\min}(B^T {\underaccent{\sim}{\boldsymbol \omega}} B)]=1.197$ and ${\boldsymbol E}[\lambda_{\max}(B^T {\underaccent{\sim}{\boldsymbol \omega}} B)]=25.35$ are computed numerically. We obtain $\kappa > 8.16$. We then adjust $\kappa=30$ and calculate the upper-bound of the sampling period which is $\tau(\kappa) <0.004$. We adjust $\tau=2$ ms and $\kappa=30$.\\ Figure \ref{p1} shows the plots of $\omega_i+n_i(\mathrm{k})$ for nodes $2,3,4$, together with the plot of relative phases and $\lambda_{\min}(B^T {\underaccent{\sim}{\boldsymbol \omega}} B)$. Notice that the weight of edge $e_{2 \rightarrow 3}= \kappa (\omega_2+n_2(\mathrm{k}))$, while $e_{3 \rightarrow 2}= e_{3 \rightarrow 4}=\kappa (\omega_3+n_3(\mathrm{k}))$ and $e_{4 \rightarrow 3}=\kappa (\omega_4+n_4(\mathrm{k}))$. As shown the relative phases stay within the desired set $S^{G}(\gamma)$ ($\gamma$ is chosen very close to $\frac{\pi}{2}$).

We then, change the mean-value of $n_3$ to $-1.6$. The result, shown in Figure \ref{p2}, is similar to Figure \ref{p1}. Notice $\omega_3=1 < |-1.6|$. In this case, ${\boldsymbol E} [\lambda_{\min}(B^T {\underaccent{\sim}{\boldsymbol \omega}} B)]$ is positive. Notice that the latter is positive despite having a negative weight on one of the edges.  

We then decrease the mean-value of $n_3$ to $-3$. The result is shown in Figure \ref{p3}. As shown, the trajectories exit the set $S^{G}(\gamma)$. For this case, ${\boldsymbol E} [\lambda_{\min}(B^T {\underaccent{\sim}{\boldsymbol \omega}} B)]$ is negative. Table I shows the value of minimum and maximum expectations of the eigenvalues of $B^T {\underaccent{\sim}{\boldsymbol \omega}} B$ for four cases: no random disturbance, Gaussian random disturbance with zero mean for all nodes (Fig. \ref{p1}), Gaussian random disturbance with zero mean for nodes $1,2,4,5$ and negative mean for node 3 (Fig.\ref{p2}, Fig. \ref{p3}). 
\begin{figure}[h]
\centering
\includegraphics[scale=0.38]{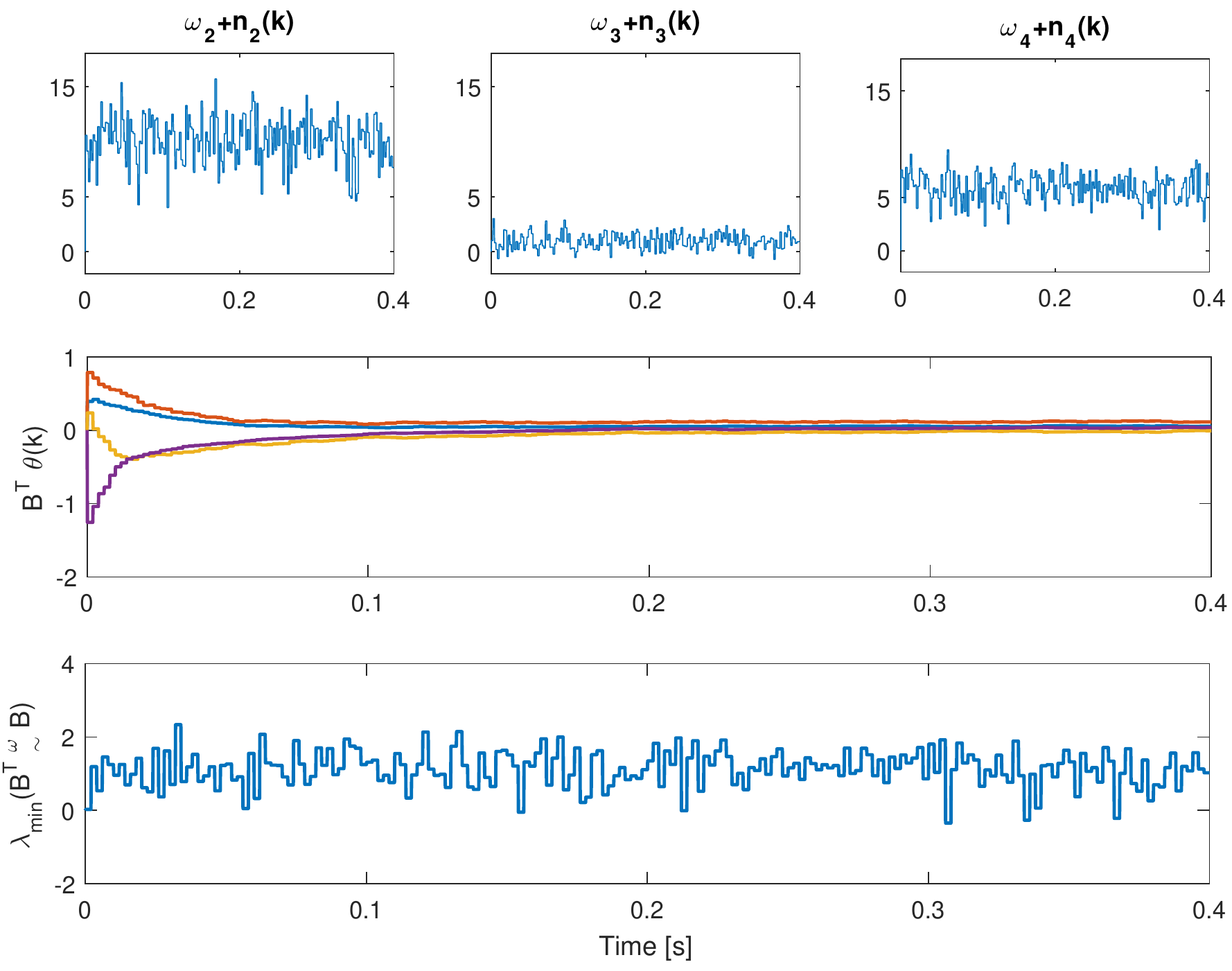}\vspace{-2mm}
\caption{Frequency-dependent network with zero-mean Gaussian noise.}\label{p1}
\end{figure}

\begin{figure}[h]
\centering
\includegraphics[scale=0.38]{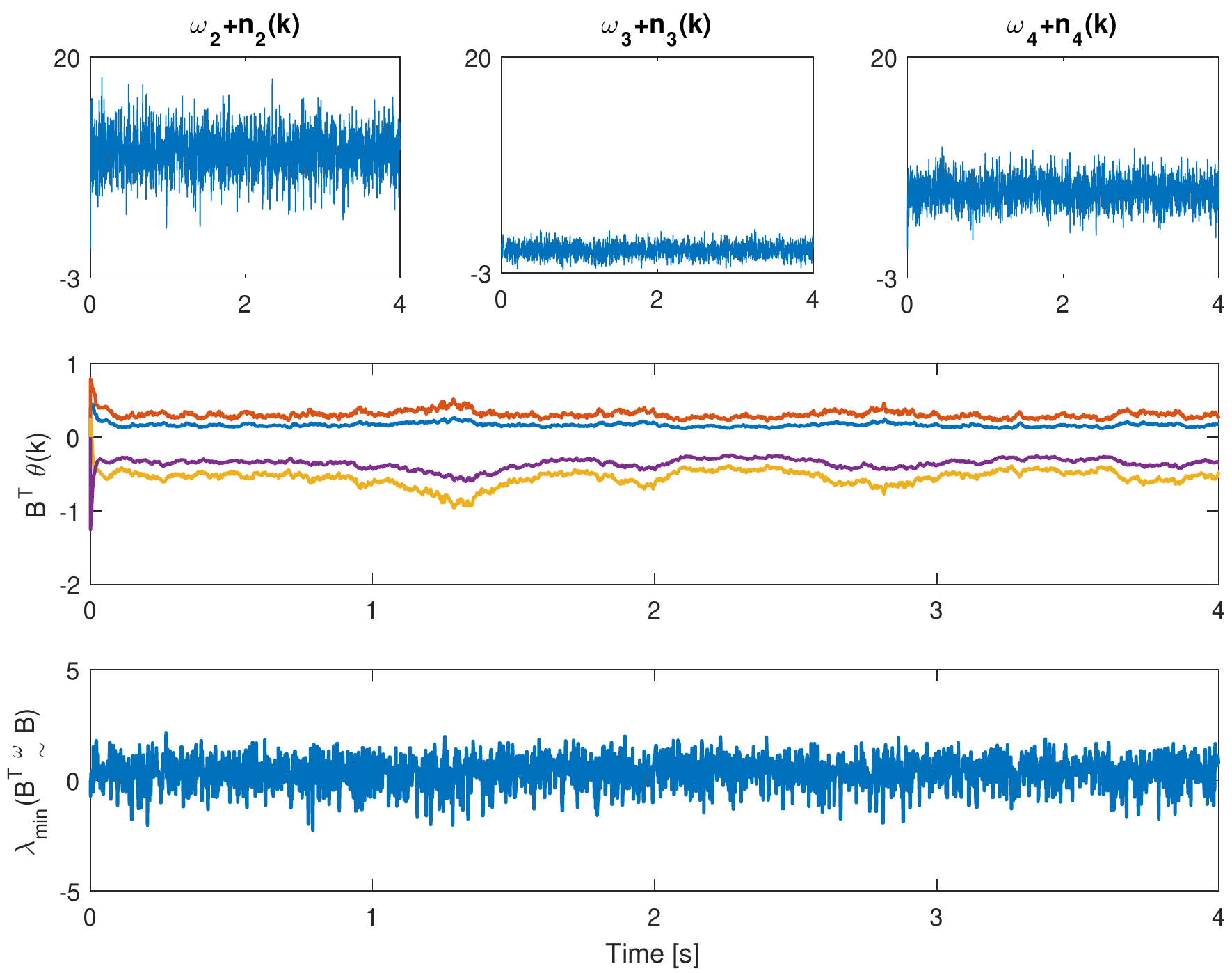}\vspace{-2mm}
\caption{Frequency-dependent network with zero-mean Gaussian noise for nodes $1,2,4,5$ and mean-value $-1.6$ for $n_3$.}\label{p2}
\end{figure}
\begin{figure}[h]
\centering
\includegraphics[scale=0.38]{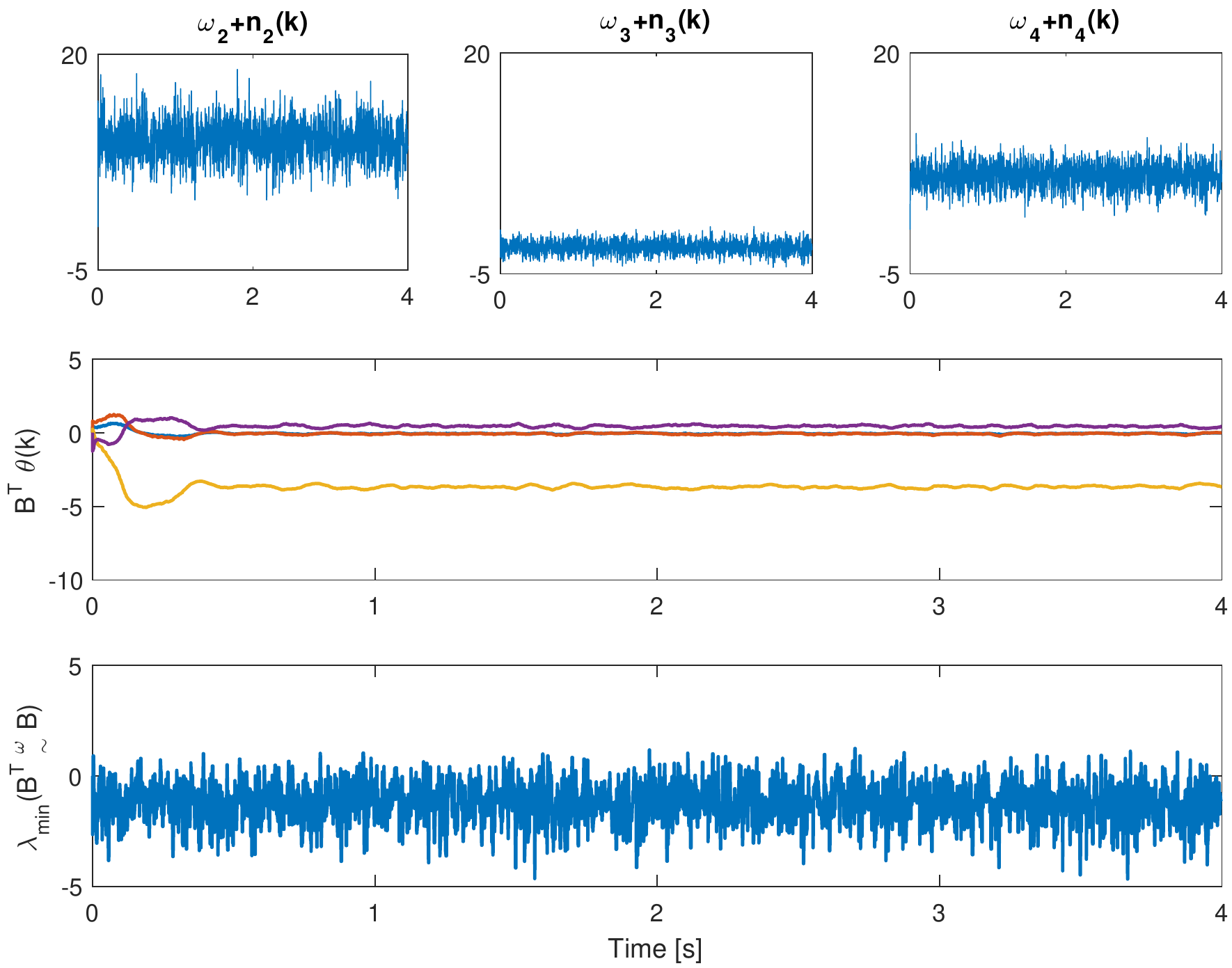}\vspace{-2mm}
\caption{Frequency-dependent network with zero-mean Gaussian noise for nodes $1,2,4,5$ and mean-value $-3$ for $n_3$.}\label{p3}
\end{figure}

\begin{center}
\begin{tabu} to 0.45\textwidth { | X[l] | X[c] | X[c] | }
\hline
{\small Choice} & ${\small {\boldsymbol E} [\lambda_{\min}(B^T {\underaccent{\sim}{\boldsymbol \omega}} B)]}$ & ${\small {\boldsymbol E} [\lambda_{\max}(B^T {\underaccent{\sim}{\boldsymbol \omega}} B)]}$\\
\hline
Noise-free & 1.31  & 24.46\\
\hline
${\boldsymbol E}[n_i]=0, \forall i$ & 1.19  & 25.35\\
\hline
${\boldsymbol E}[n_3]=-1.6$ & 0.34  & 24.05\\
\hline
${\boldsymbol E}[n_3]=-3$ & -1.17  & 23.669\\
\hline
\end{tabu}\\[1mm]
\small{Table 1}\label{t1}
\end{center}

Finally, we examine the result of Corollary \ref{co2}, and simulate a network as shown in Fig. \ref{g2} with constant weights. We consider Gaussian distributions with zero mean for nodes $1,4,5$, negative mean values $-20,-2$ for nodes $2,3$, respectively. The variances, $\kappa$ and $\tau$ are set as before. Notice that negative mean values do not change the result as long as the size of variations is considered in the calculation of $\kappa$ and $\tau$.
\begin{figure}[h]
\centering
\includegraphics[scale=0.5]{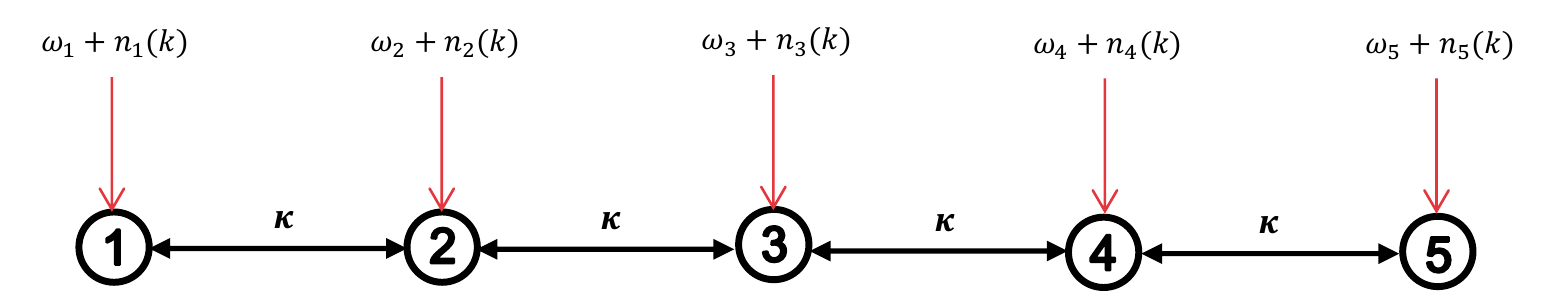}
\caption{Undirected line network.}\label{g2}
\end{figure}
\begin{figure}[h!]
\centering
\includegraphics[scale=0.38]{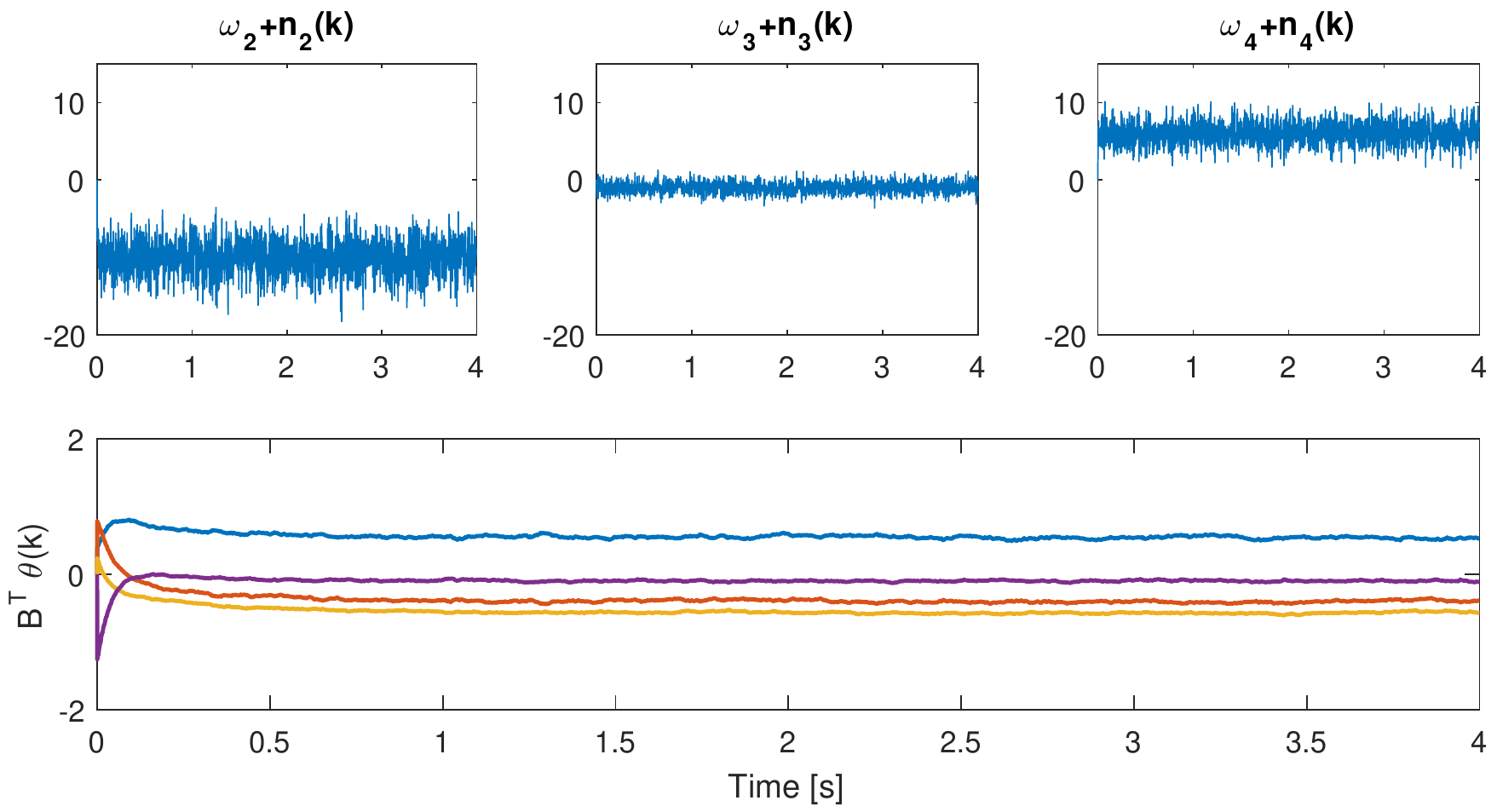}\vspace{-2mm}
\caption{Exogenous frequencies and relative phases for the undirected network.}
\label{p6}
\end{figure}
\section{Conclusions}\label{sec:cl}
This paper studied conditions under which stochastic phase-cohesiveness is achieved for a number of Kuramoto oscillators in a bidirectional network and an undirected network. The exogenous frequencies have been assumed to be combined with random disturbances representing uncertainties. For the bidirectional network, a sufficient coupling and a necessary sampling-time conditions have been obtained in order to prove a recurrent property for the process provided that the expected value of the minimum eigenvalue of the random weighted edge Laplacian is strictly positive. Similar results have been proved for the undirected network where no assumption on the weighted edge Laplacian is required. 
\bibliographystyle{ieeetr}
\bibliography{biblio}
\end{document}